\documentclass[runningheads]{llncs}

\usepackage{amsthm}
\usepackage[T1]{fontenc}
\usepackage{cite}
\usepackage{amsmath,amsthm,amssymb,amsfonts}
\usepackage{amsmath}  
\usepackage{float}
\usepackage[ruled,vlined,boxed]{algorithm2e}

\newtheorem{define}{Definition}
\usepackage{graphicx}
\usepackage{subfigure}
\usepackage{makecell}
\usepackage{enumerate}
\usepackage{multirow}
\usepackage{color}
\usepackage{slashed}
\usepackage{bbding}
\usepackage{booktabs}
\usepackage{subfigure}
\usepackage{pifont}
\usepackage{url}
\usepackage{diagbox}
\usepackage{threeparttable}
\usepackage[marginal]{footmisc}

\usepackage[utf8]{inputenc}
\begin{document}

\title{\emph{SecGraph}: Towards SGX-based Efficient and Confidentiality-Preserving Graph Search \\(full version)\thanks{This paper has been accepted by DASFAA 2024.}}

\author{Qiuhao Wang\and
Xu Yang\and
Saiyu Qi\inst{(}\Envelope\inst{)} \and
Yong Qi}

\authorrunning{Wang. et al.}
\institute{Xi’an Jiaotong University, Xi’an, Shaanxi, China \\
\email{\{QiuHaoWang,yangxu\}@stu,xjtu.edu.cn,\{saiyu-qi,qiy\}@xjtu.edu.cn}}

\maketitle              % typeset the header of the contribution

\footnote{The first two authors contributed equally and share the “co-first author” status. }

\begin{abstract}
Graphs have more expressive power and are widely researched in various search demand scenarios, compared with traditional relational and XML models. Today, many graph search services have been deployed on a third-party server, which can alleviate users from the burdens of maintaining large-scale graphs and huge computation costs. Nevertheless, outsourcing graph search services to the third-party server may invade users’ privacy. PeGraph was recently proposed to achieve the encrypted search over the social graph. The main idea of PeGraph is to maintain two data structures \emph{XSet} and \emph{TSet} motivated by the OXT technology to support encrypted conductive search. However, PeGraph still has some limitations. First, PeGraph suffers from high communication and computation costs in search operations. Second, PeGraph cannot support encrypted search over dynamic graphs. In this paper, we propose an SGX-based efficient and confidentiality-preserving graph search scheme \emph{SecGraph} that can support insertion and deletion operations. We first design a new proxy-token generation method to reduce the communication cost. Then, we design an LDCF-encoded \emph{XSet} based on the Logarithmic Dynamic Cuckoo Filter to reduce the computation cost. Finally, we design a new dynamic version of \emph{TSet} named \emph{Twin-TSet} to enable encrypted search over dynamic graphs. We have demonstrated the confidentiality preservation property of \emph{SecGraph} through rigorous security analysis. Experiment results show that \mbox{\emph{SecGraph}} yields up to 208$\times$ improvement in search time compared with \mbox{PeGraph} and the communication cost in PeGraph is up to 540$\times$ larger than that in \emph{SecGraph}.

% Due to the advancement of cloud computing, the practice of outsourcing social graph data to public servers for the provision of high-quality online social network (OSN) services\cite{DBLP:journals/pvldb/CurtissBBDGJKLPSSWYZ13} has gained significant popularity. It's worth noting that cloud data is susceptible to attacks, and some research has been devoted to mitigating the challenge. However, the current research has limitations in dealing with dynamic data, while there are bottlenecks in system performance. To address this, we propose a novel \underline{S}GX-based \underline{e}fficient and \underline{c}onfidentiality-preserving \underline{graph} search (\emph{SecGraph}) scheme that is the first secure social graph search scheme that supports secure dynamic graph updates while offering rich search categories and improving search efficiency and security. \emph{SecGraph} leverage the cryptographic algorithm and trusted hardware to accelerate the update and search process while preventing private leakage. We implement \emph{SecGraph} and conduct experiments on multiple real-world datasets with varying sizes and distributions. The experimental results show \emph{SecGraph} improves search performance from $10 \times$ to $3,000 \times$.
\end{abstract}

\section{Introduction}
Graphs are gaining increasing attention since they have expressive power and play an important role in many applications, such as social networks \cite{DBLP:conf/dasfaa/WangZYSX21}, biological data analyses \cite{DBLP:conf/ijcai/ZhaoQYZGLL23}, recommender systems \cite{DBLP:conf/wsdm/GaoW0022}, etc. Essentially, this is because the core data involved in these applications can be conveniently represented as graphs. For example, social networks (e.g., Facebook \cite{Facebook} and Instagram \cite{Instagram}) constitute various social users, essentially a graph where vertices represent users and edges represent their relationships, such as friendship. The wide use of graphs has brought about the emergence of graph search, such as shortest path search \cite{DBLP:conf/asiaccs/GhoshKT21} and neighbor query \cite{DBLP:conf/dasfaa/ZhangTSZXCZXLZ23}, etc. As the scale of the graph surges, graph owners (e.g., enterprises for graph based services) desire to outsource their graphs to a third-party server, which however may snoop sensitive user information, e.g., users' social connections, interests, and potentially sensitive information.

To enable encrypted graph search, some solutions rely on structure encryption \cite{DBLP:conf/ndss/CashJJJKRS14} to enable private adjacent vertices search over encrypted graph \cite{DBLP:conf/asiacrypt/ChaseK10,DBLP:conf/trustcom/TengCSB16,DBLP:conf/ccs/LaiYSLLL19}. However, these solutions just support single keyword search (i.e., search for the neighbor vertices of a vertex). Recently, Wang \textit{et al.}\cite{DBLP:journals/tifs/WangZJY22} proposed PeGraph based on the OXT technology \cite{DBLP:conf/crypto/CashJJKRS13} to support encrypted conductive search over social graphs, i.e., search for common neighboring vertices (e.g., friends) of multiple vertices (e.g., people). Specifically, PeGraph maintains a multimap data structure \emph{TSet} to store the encrypted form of mapping from vertices to all their neighbor vertices and a set \emph{XSet} to store pairs of adjacent vertices information at the server side. Upon receiving search keywords (e.g., $w_1 \wedge ... \wedge w_n$)\footnote{In this paper, we assume that the search keywords count issued by a client is $n$ in the form of $w_1 \wedge ... \wedge w_n$ and $w_1$ is the least frequent unless otherwise specified.}, the server first finds the matching neighbor vertices of $w_1$ in \emph{TSet} as an initial search result. Then, for each of its neighbor vertex, the server checks the existence of all pairs where each one is made up by it with a remaining search keyword $w_i, i\in[2,n],$ in \emph{XSet} to determine whether to insert it to the final result.

However, PeGraph suffers from three limitations. First, PeGraph incurs high communication costs in a search operation since it requires two search roundtrips between a client and the server, due to the characteristic of the OXT technology \cite{DBLP:conf/crypto/CashJJKRS13}. Second, PeGraph incurs high computation cost since it requires $c\cdot (n\text{-}1)$ expensive exponentiation modulo operations to check the existence of each neighbor vertex in the initial search result, where $c$ is the number of neighbor vertices in the initial search result and $n$ denotes the search keywords count. Last, PeGraph cannot support encrypted search over dynamic graphs since the OXT technology is designed to be applied to static encrypted search. 

In this paper, we propose an \emph{SGX-based efficient and confidentiality-preserving graph search scheme (SecGraph)} to provide secure and efficient search services over encrypted dynamic graphs. To address the first limitation, we design a \emph{proxy-token generation method} by exploiting the trusted hardware SGX as the client's trusted proxy on the server side to reduce the communication cost of search operations. To address the second limitation, we design a novel data structure \emph{LDCF-encoded XSet} based on the Logarithmic Dynamic Cuckoo Filter \cite{DBLP:conf/icde/ZhangC0R21} to transform the expensive exponentiation modulo operations required for existence checking operations in the OXT technology into the LDCF-based membership check process which can be conducted within storage-constraint SGX to reduce the computation cost. To address the last limitation, we design a new dynamic version of \emph{TSet} named \emph{Twin-TSet} to record the relationships between adjacent vertices to efficiently handle encrypted searches over dynamic graphs. In summary, the contributions are the following.

\begin{enumerate}

\item[$\bullet$] We propose an SGX-based efficient and confidentiality-preserving graph search scheme \emph{SecGraph} to provide secure and efficient graph search services.

\item[$\bullet$] In \emph{SecGraph}, we design a new proxy-token generation method to reduce the communication cost. Then, we design an LDCF-encoded \emph{XSet} to reduce the computation cost. Moreover, we design a new dynamic version of \emph{TSet} named \emph{Twin-TSet} to enable encrypted search over dynamic graphs.

\item [$\bullet$] We further extend \emph{SecGraph} into two optimized schemes \emph{SecGraph-G} and \emph{SecGraph-P} by adopting the fingerprint grouping and checking parallelization strategies, respectively, to speed up the search procedure.

\item[$\bullet$] Finally, experiment results show that \mbox{\emph{SecGraph}}, \emph{SecGraph-G} and \mbox{\emph{SecGraph-P}} respectively yield up to 208$\times$, 572$\times$ and 3,331$\times$ improvements in search time compared with the state-of-the-art scheme PeGraph, and the communication cost in PeGraph is up to 540$\times$ larger than that in \emph{SecGraph}.
\end{enumerate}

\section{Related Work}

Graphs are widely used to model structure data in various domains (e.g., social networks \cite{DBLP:conf/dasfaa/WangZYSX21}, and protein structures \cite{DBLP:conf/ijcai/ZhaoQYZGLL23}, and more) due to their excellent capability in characterizing the complex interconnections among entities. Graph search, as one of the fundamental tasks in graph analytics, has gained increasing attention in recent years. Numerous algorithms have been proposed to address diverse types of graph search, such as graph similarity search\cite{DBLP:journals/tkde/ChenHHVZZ21} and shortest path search\cite{DBLP:conf/asiaccs/GhoshKT21}, etc. However, all of them focus on graph search in the plaintext domain without considering privacy preservation.

Along with the development of the data outsourcing paradigm, enabling efficient search over encrypted outsourced graphs has attracted widespread attention from both academia and industry. Structure encryption \cite{DBLP:conf/ndss/CashJJJKRS14} is a promising cryptographic technique to provide private adjacent vertices search over encrypted graphs. For example, Chase \textit{et al.} \cite{DBLP:conf/asiacrypt/ChaseK10} firstly extended the notion of structured encryption to the setting of arbitrarily structured data including complex graph data. Based on this, Shen \textit{et al.} \cite{DBLP:conf/trustcom/TengCSB16} proposed a privacy-preserving scheme PP$k$NK to achieve top-$k$ nearest keyword search on graphs. Furthermore, Lai \textit{et al.} \cite{DBLP:conf/ccs/LaiYSLLL19} proposed a privacy-preserving graph scheme GraphSE$^2$ to facilitate parallel and encrypted graph data access on large social graphs. Recently, Wang \textit{et al.}\cite{DBLP:journals/tifs/WangZJY22} proposed PeGraph using the OXT technique \cite{DBLP:conf/crypto/CashJJKRS13} and a more lightweight secure multi-party computation method to expand the richness of search and improve performance. PeGraph is state-of-the-art privacy-preserving social graph search scheme. However, there are some drawbacks. First, PeGraph incurs high communication costs since it requires two search roundtrips between a client and the server during the search procedure. Second, PeGraph suffers from high computation costs since it requires expensive exponentiation modulo operations to check the existence of each entry in the initial search result. Last, PeGraph cannot support encrypted search over dynamic graphs since the OXT technology is applied to static encrypted search. In this paper, we devote ourselves to achieve an efficient and confidentiality-preserving graph search scheme while supporting dynamic graph updates. 

\section{Preliminaries}
\textbf{Intel SGX and Enclave.} Intel SGX \cite{cryptoeprint:2016/086} is a set of extensions of x86 instructions that provides trusted execution environments (i.e., \emph{enclave}) to protect the integrity and confidentiality of the application data and the code. The enclave is limited to 128MB to store data. If this limit is exceeded, the enclave will automatically apply the page-swapping mechanism, causing severe performance degradation. The enclave has three main security properties: (1) \emph{isolation:} any software outside the enclave can not directly access code or data within it, but it can access the entire virtual memory of its untrusted host; (2) \emph{sealing:} it enables the process of encrypting and authenticating the data within it; (3) \emph{attestation:} there is a secure channel between an external party and the enclave.

\noindent\textbf{Oblivious Cross-Tags Protocol.} 
Oblivious cross-tags protocol (OXT) \cite{DBLP:conf/crypto/CashJJKRS13} is a searchable encryption technique designed for secure and efficient conjunctive search over text files. Generally, the core idea of OXT is maintaining a multimap data structure \emph{TSet} and a set data structure \emph{XSet} at the server side, \emph{TSet} records a set of pairs ($fid,y$) for each keyword $w$ labeled by a corresponding value $stag(w)$ where $fid$ is the encrypted file identifier and $y$ is a blinded value, \emph{XSet} records a list of values $xtag(w, fid)$ over each keyword-file-identifier pair. Specifically, given a pseudo-random function (PRF) $F$ and the keys ($K_I,K_X,K_T,K_Z$), we assume there are $c$ files $\{fid_i\}, i \in [1,c]$ containing the keyword $w$. The client sequentially calculates $stag(w) = F(K_T,w)$, $y=F(K_I,fid_i) \cdot F(K_Z,w||i)^{-1}$, and $xtag(w,fid_i) = g^{F(K_X,w) \cdot F(K_I,fid_i)},i \in [1,c]$, where $g$ is the generator of a cyclic group and $\text{'||'}$ represents the concatenation operation. When issuing a conjunctive search (e.g., $w_1 \wedge ... \wedge w_n$), the client first sends the search token $stag(w_1)$ to retrieve all matching encrypted $fid$ containing $w_1$ from the server. With the size $c$ of the initial search result just retrieved, the client then generates and sends the intersection tokens $xtoken(w_1,w_i,j) = g^{F(K_Z,w_1||j) \cdot F(K_X,w_i)}, i\in [2,n],j \in [1,c]$ to the server. Upon receiving them, for each encrypted $fid$ in the initial search result, the server will determine whether to insert it into the final search result by checking $xtoken(w_1,w_i,j)^y = g^{F(K_X,w_i) \cdot F(K_I,fid_j)} \overset{\text{?}}{=} xtag(w_i,id_j), i\in [2,n]$ for each $j \in [1,c]$. Using the OXT technology, the server can exactly find the search result without knowing either file identifier $fid$ or keywords $w_i,i\in [2,n]$.

% Oblivious Dynamic Cross-Tags Protocol (ODXT)\cite{DBLP:conf/ndss/PatranabisM21} is an enhanced version of OXT to support updates to dynamic conjunctive SSE with forward and backward privacy. To support dynamism, ODXT has the following three major improvements. First, ODXT's client stores a map \textsf{UpdateCnt}$[w]$ to record each $w$'s update count $c$. Second, \emph{TSet} stores a mapping from $(w,i)_{1\le i \le c}$ to the corresponding $fid$ and operation type $op$ (insert or del) for this update. $xtag$ stored in \emph{XSet} is adjusted to $h(w, fid, op)$. Last, it supports lazy deletion, that is adding a delete type of update in the update phase, and the client will autonomously delete the deleted \textit{fid} from the results in the search phase.
\begin{figure}[t]
    	\centering
    	\includegraphics[width=1\linewidth]{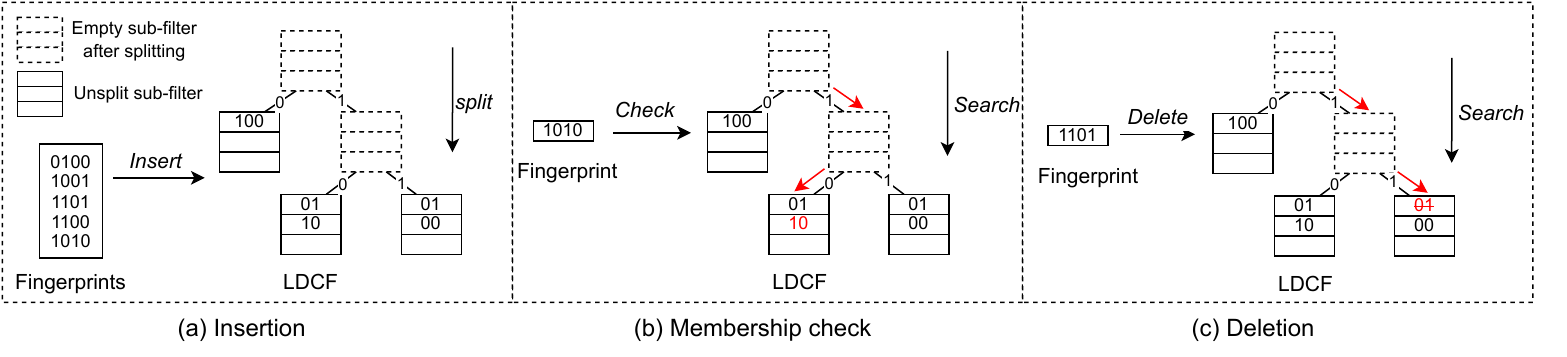}
    	\caption{Illustration of the data structure LDCF and three associated operations.}
    	\label{fig:LDCF}
\vspace{-1.8em}
\end{figure}

\noindent\textbf{Cuckoo Filter and Dynamic Cuckoo Filter.}
Cuckoo Filter (CF) \cite{DBLP:conf/conext/FanAKM14} is a compact data structure that enables approximate set membership checks in static settings. The main idea of the CF is each element is hashed to one or more buckets, and these buckets store fingerprint information derived from the hashed values, where the fingerprint is a small portion of the original hash value, typically a few bits. To support efficient inserting, deleting, and checking an item in the CF, Zhang \textit{et al.} \cite{DBLP:conf/icde/ZhangC0R21} extend the traditional CF to a novel data structure Logarithmic Dynamic Cuckoo Filter (LDCF) where a fully inserted CF is divided into two CFs in a binary tree shape recursively based on the prefix of the fingerprints, as shown in Fig. \ref{fig:LDCF}. Due to the features of the tree structure, LDCF achieves a sub-linear search and reduces space storage overhead by omitting the storage space of common prefixes for the fingerprints during the process of tree expansion. For convenience, we refer to each CF in LDCF as a \emph{sub-filter} in this paper. LDCF consists of three algorithms:  

\begin{enumerate}
\item[$-$] LDCF.\textsf{Insert}$(x)$: Upon receiving an inserting element $x$, the algorithm first computes its fingerprint $\delta$ and uses it to locate the matching \emph{sub-filter}. Then it calculates two candidate positions $\mu$ and $\nu$ and finally puts $\delta$ in one of the two candidate positions in the matching \emph{sub-filter}. Fig. \ref{fig:LDCF} (a) shows the procedure of inserting five fingerprints into an LDCF. 

% \item[$-$] LDCF.\textsf{Append}$(loc, x)$: The algorithm takes the current insertion position $loc$ and the level $l$ as input, finds the \emph{sub-filter} that needs to be split and decides which child-CF to insert according to the $\delta_x$'s first bit. The first bit of $\delta_x$ will be removed after insertion.

\item[$-$] LDCF.\textsf{Membershipcheck}$(x)$: Upon receiving a query element $x$, the algorithm first calculates its fingerprint $\delta$ and finds the matching \emph{sub-filter}. Then, it derives the corresponding positions $\mu$ and $\nu$ to check whether $\delta$ exists in one of them and returns true if exists, otherwise false. Fig. \ref{fig:LDCF} (b) shows the procedure of checking a fingerprint whether exists in the LDCF. 

\item[$-$] LDCF.\textsf{Delete}$(x)$: Upon receiving a deleting element $x$, the algorithm first calculates its fingerprint $\delta$ and finds the matching \emph{sub-filter}. Then, it derives the corresponding positions $\mu$ and $\nu$ to locate $\delta$, and finally deletes it and returns true, otherwise false. Fig. \ref{fig:LDCF} (c) shows the procedure of deleting an existing fingerprint in the LDCF. 
\end{enumerate}

\section{System Overview}
\subsection{Problem Definition}

\begin{figure}[t]
    	\centering
    	\includegraphics[width=0.7\linewidth]{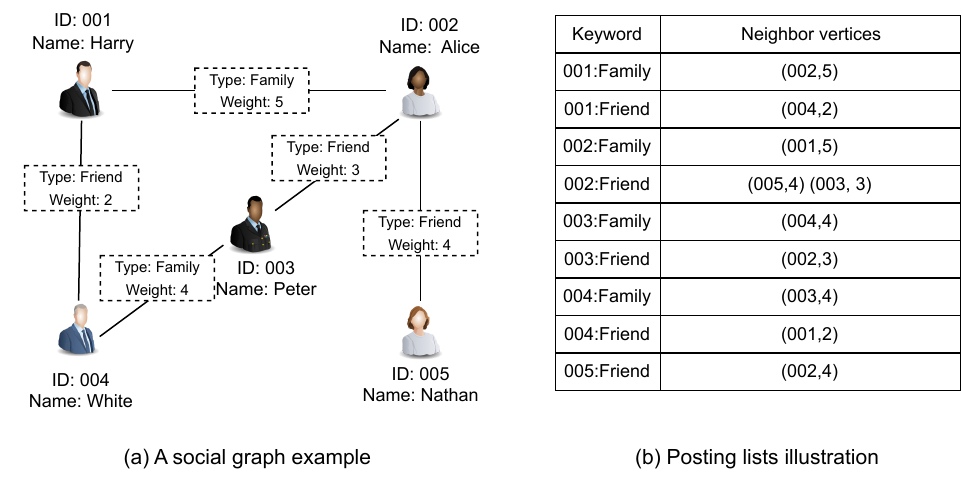}
    	\caption{A toy example for illustrating a social graph.}
    	\label{fig:fig1}
\vspace{-1.8em}
\end{figure}

Formally, the social graph can be represented by $G = (V, E)$, where $V$ is a vertex set in the graph and $E$ is a relationship set among vertices. As shown in Fig. \ref{fig:fig1}(a), each $v \in V$ represents a vertex identified by $id$, it also has an attribute \emph{name}, e.g., the vertex '001' named 'Harry'. Each $e \in E$ represents an edge between vertices identified by $type$, with $weight$ representing the importance of the edge, e.g., there is a friendship relationship between vertices '001' and '002', with an importance of 5. A relationship can be defined by $(id_{out}, id_{in}, type)$, which means the edge labeled by $type$ starts from $id_{out}$ and ends with $id_{in}$. As shown in Fig. \ref{fig:fig1}(b), our system maintains a posting list for each vertex. The posting list is a map data structure, where the keyword is $id_{out}:type$ and the value is composed of a set of ($id_{in}, weight)$ pairs, where $id_{in}$ is the identifier of the neighbor vertex of $id_{out}$, and $weight$ represents the importance of the relationship, e.g., the intimacy between friends. In this paper, we denote $id_{out}$ as $w$, and the search types on graphs are divided into two types: (1) \emph{Exact search.} Search for neighbor vertices with a specific relationship to multiple vertices, e.g., the common friend of vertices '003' and '005' is vertex '002'. (2) \emph{Fuzzy search.} Search for vertices whose \emph{name} contain the same sub-string, e.g., the vertices with 'ha' in their \emph{name} include vertices '001' and '005'.
% \item \textit{Connectivity Search:} Users may want to find the connectivity information in the graph, the Connectivity search has two types: 
%     \begin{itemize}
%         \item[1)] \textit{Outbound Search:} return all destination nodes that the source node can reach within $K$ steps while the edges on the path belong to a certain type. \textit{E.g.,} search 002's friend within two steps, return \{004\}.
%         \item[2)] \textit{Inbound Search:} return all destination nodes that can reach the source node within $K$ steps while the edges on the path belong to a certain type.
%     \end{itemize}
% \item \textit{Specific Search:} Perform sorting, top-k, or other complex calculations on search results. \textit{E.g.,} search 004's friends, and return in order 

\subsection{System Model}
There are two entities in the system: the client and the SGX-enabled server. (1) \emph{Client.} The client launches a remote attestation and establishes a secure channel with the trusted part (i.e., enclave) of the server, then sends a set of secret keys to the enclave. The client can securely insert or delete edges or vertices, and issue conjunctive search. (2) \emph{SGX-enabled server:} The server will insert or delete the corresponding encrypted index and provide graph search services for the client. We assume that the client is absolutely honest but the server is honest-but-curious. That is, the server will execute the established program correctly, but powerful enough to get full access over the software stack (such as OS and hypervisor) outside of the enclave, and can infer sensitive information from encrypted data by observing search tokens and search results. We state that a series of side-channel attacks against SGX are out of our scope.

% \subsection{Design Goal}
In this paper, \emph{SecGraph} aims to achieve three goals in terms of efficiency, functionality, and security : (1) \emph{Efficient search.} \emph{SecGraph} should provide efficient search while reducing communication and computation costs. (2) \emph{Provision dynamic update.} \emph{SecGraph} should support encrypted search over dynamic graphs. (3) \emph{Confidentiality preservation.} \emph{SecGraph} should protect the confidentiality of search keywords and results.

% \begin{itemize}
% \item[$-$] \textit{Efficient dynamic update:} \emph{SecGraph} should support efficient dynamic graph updates.
% \item[$-$] \textit{A variety of efficient search types:} \emph{SecGraph} should achieve abundant encrypted graph data search mentioned in Section 4.1 with high efficiency.
% \item[$-$] \textit{Confidentiality Preservation:} \emph{SecGraph} should process in encryption, and ensure forward and backward security, in the meantime avoiding search pattern leakage.
% \end{itemize}

\section{Detailed Construction of \emph{SecGraph}}
% In this section, we first introduce the high-level description of \emph{SecGraph}. Then, we illustrate the detail construction of \emph{SecGraph}. 

\subsection{Design Rationale of \emph{SecGraph}}

As mentioned, Wang \textit{et al.}\cite{DBLP:journals/tifs/WangZJY22} proposed PeGraph that can support versatile search over encrypted social graphs. The main idea of PeGraph is combining the OXT technology \cite{DBLP:conf/crypto/CashJJKRS13} and the additive secret sharing\cite{DBLP:conf/ndss/Demmler0Z15} to provide private ranked conjunctive search. Specifically, PeGraph first extracts each keyword $w$ in the form of $id_{out}:type$ from graphs and associates the posting list $\{(id_{in}, weight)\}$ with $w$. Then, for each $w$, PeGraph calculates $stag(w)$ to index its corresponding posting list. In addition, for each $id_{in}$ in the set $\{(id_{in},weight)\}$, PeGraph calculates $xtag(w,id_{in})$ and adds it to \emph{XSet}, which allows the server to check the existence of an $id_{in}$ without knowing either $w$ or $id_{in}$. After obtaining the matching $id_{in}$, PeGraph further allows the server to obliviously render the encrypted search results in a ranked order according to their importance $weight$ by using the additive secret sharing primitive and returning the top-$k$ results, avoiding unnecessary downlink traffic. 

Despite its considerable performance, PeGraph still has some problems due to the characteristics of OXT technology \cite{DBLP:conf/crypto/CashJJKRS13}. First, PeGraph incurs high communication costs since it requires two search roundtrips between a client and the server during the search procedure. The reason is that the OXT technology requires the client to retrieve the encrypted matching posting list associated with a certain keyword (e.g., $w_1$) as an initial search result from the server in advance to assist in generating intersection tokens (e.g., $xtoken(w_1,w_i,j), i\in [2,n], j\in[1,c]$). Therefore, the first challenge is how to achieve one search roundtrip to reduce the communication cost. Second, PeGraph suffers from high computation cost since it requires $c\cdot (n\text{-}1)$ expensive exponentiation modulo operations to check the existence of each $id_{in}$ in the initial search result, where $c$ is the number of $id_{in}$ in the initial search result and $n$ is the number of search keywords. Therefore, the second challenge is how to check the existence of $id_{in}$ without using expensive exponentiation modulo operations to reduce computation cost. Finally, PeGraph cannot support encrypted search over dynamic graphs since the OXT technology only supports static encrypted search. Accordingly, the third challenge is how to enable dynamic \mbox{encrypted graph search.}

To address the first challenge, we propose a \emph{proxy-token generation method} to avoid the procedure of returning the initial search result to the client. The main idea is to resort to SGX as the client's trusted proxy at the server side to generate a search token $stag(w_1)$ and intersection tokens $xtoken(w_1,w_i,j), i\in [2,n], j\in[1,c]$ due to the trusted computing power of SGX. Specifically, when issuing a graph search, a client sends the search keywords $w_1 \wedge ... \wedge w_n$ to the enclave (i.e., the trusted part of the server) via a secure channel, and the latter generates $stag(w_1)$ and retrieves the initial search result from the untrusted part of the server to generate intersection tokens independently. In this way, only one search roundtrip is required between the client and the server, which largely reduces the communication cost.

To address the second challenge, we start with an observation from the OXT technology \cite{DBLP:conf/crypto/CashJJKRS13}. The server needs to perform a large number of expensive exponentiation modulo operations to check whether each $xtag(w_i,id_{in}), i\in[2,n]$ is in \emph{XSet} or not to determine whether $id_{in}$ belongs to the final search result, which can incur significant performance degradation. Fortunately, we note that the procedure of checking the existence of each $id_{in}$ in the initial search result is essentially a membership check process. Thus, a naive solution is to load the data structure \emph{XSet} into the enclave due to the trusted storage power of SGX and allow the enclave to check whether each $xtag(w_i,id_{in}), i\in[2,n]$ exists in \emph{XSet} in plaintext rather than perform exponentiation modulo operation over ciphertexts. However, as the entire keywords in graphs are encoded into \emph{XSet}, the size of \emph{XSet} will become large and may not be stored in the enclave due to its limited storage, which will incur the page-swapping mechanism and also cause performance degradation. To do this, we use a compact data structure Logarithmic Dynamic Cuckoo Filter (LDCF) \cite{DBLP:conf/icde/ZhangC0R21} to store \emph{XSet}. The \emph{LDCF-encoded} \emph{XSet} not only allows the enclave to load a few \emph{sub-filters} to complete the procedure of checking the existence of each $id_{in}$ in the initial search result instead of loading the entire \emph{sub-filter} set but provides a sub-linear search to locate the matching \emph{sub-filters} quickly. In this way, using the LDCF-encoded \emph{XSet}, the computation cost is largely reduced during the search procedure.

To address the last challenge, we note that applying the OXT technology \cite{DBLP:conf/crypto/CashJJKRS13} to the dynamic setting is not straightforward since the pre-computed multimap data structure \emph{TSet} does not need to change with time, thus it is impossible to emulate in the dynamic setting, where the graphs are continuously updated (including insertion and deletion). To do this, we design a new dynamic version of \emph{TSet} named \emph{Twin-TSet} that contains a pair of map data structures, one is \emph{TSet} storing the relationship between a pair made up by a keyword $w$ and an associated update count $c$ and a neighbor vertex $id_{in}$ the other is \emph{ITSet} (i.e., inverse \emph{TSet}) storing the relationship between a pair made up by a keyword $w$ and a neighbor vertex $id_{in}$ and the corresponding update count $c$. When inserting a pair $(w,id_{in})$, we first increase the update count $c$ associated with $w$ by one and then insert the two pairs $((w,c),id_{in})$ and $((w,id_{in}),c)$ into \emph{TSet} and \emph{ITSet} in the form of ciphertext respectively. When deleting a pair $(w,id_{in})$, we first retrieve the corresponding update count $c$ from \emph{ITSet} and then replace the $id_{in}$ at position $(w,c)$ of \emph{Tset} with the latest $id'_{in}$ at position $(w,c')$, where $c'$ is the latest update count associated with $w$. After that, we accordingly insert a new pair $((w,id'_{in}),c)$ into \emph{ITSet} and decrease $c$ by one. when searching $w$, similarly with the OXT technology, we only need to retrieve $c$ values from \emph{TSet} at positions $(w,i), i\in[1,c]$ without considering deleted identifiers $id_{in}$ since they have been deleted previously. At present, it seems that the \emph{Twin-TSet} can support encrypted search on dynamic graphs, but how to store the update count for each $w$ privately is still an issue. Fortunately, we can protect it from being leaked to the server by storing a map data structure \emph{UpdateCnt} within the enclave with its trusted storage power.

Until now, we have illustrated the design rationale of the proposed SGX-based efficient and confidentiality-preserving graph search (\emph{SecGraph}). It is worth noting that, due to the trusted computing power of SGX, \emph{SecGraph} can also support returning top-$k$ ranked search results to the client by directly decrypting and sorting the neighbor vertice identifiers according to their $weight$ without performing additive secret sharing operations over the encrypted posting lists.

\subsection{Design Details of \emph{SecGraph}}
Now we proceed to describe our construction of \emph{SecGraph} in detail. \emph{SecGraph} utilizes a set of PRFs ($F_1,F_2,F_3:\{0,1\}^\lambda \times \{0,1\}^* \rightarrow \{0,1\}^\lambda $), two hash functions $H_1: \{0,1\}^* \rightarrow \{0,1\}^\alpha$ and $H_2: \{0,1\}^* \rightarrow \{0,1\}^\xi$. The workflow of \emph{SecGraph} can be divided into the following protocols.

\setlength{\textfloatsep}{0.1em}
\begin{algorithm}[t]\scriptsize
\caption{\textsf{Update}}
\label{alg:alg2}
\SetNlSty{text}{}{:}
\SetAlgoNoLine
\SetAlgoNoEnd
% \SetAlgoVlined
% \LinesNumbered
\begin{multicols}{2}
\KwIn{Secrity keys $K_T, K_Z, K_X$, update counter \emph{UpdateCnt}, updated triplet $(w, id, weight)$, update operation $op$, an \emph{IndexTree}.}
\KwOut{Encrypted database \emph{EDB}=(\emph{TSet, ITSet, XSet}).}
\underline{Client:}\\
\nl Send $(w, id, weight),op$ to the enclave;\\
\underline{Enclave:}\\
\nl $xtag = (w||id),\delta = H_2(xtag),\mu = H_1(xtag)$;\\
% \nl $\nu = \mu \oplus hash(\delta)$;\\
\nl Get \emph{sub-filterId} according to $\delta$ and \emph{IndexTree};

\nl \textbf{if} $op=insert$ \textbf{then}\\
\nl ~~~ \textbf{if} $\emph{UpdateCnt}[w] = \perp$ \textbf{then}\\
\nl ~~~~~~~~$\emph{UpdateCnt}[w] = 0$;\\
\nl ~~~ $\emph{UpdateCnt}[w] \text{++} $;\\
\nl ~~~ $stag = \mathrm{\textsf{$F_1$}}(K_T,w|| \emph{UpdateCnt}[w])$;\\
%\nl $istag = \mathrm{\textsf{$F_1$}}(K_I,w|| {\mathrm{\textsf{UpdateCnt}}[w]})$;\\
\nl ~~~ \textcolor{black}{$C_{id} = (id||weight) \oplus \mathrm{\textsf{$F_2$}}(K_Z,w)$};\\
\nl ~~~ $ind = \mathrm{\textsf{$F_3$}}(K_X,w||id)$;\\
\nl ~~~ \textcolor{black}{$C_{stag} = (w|| \emph{UpdateCnt}[w]) \oplus \mathrm{\textsf{$F_2$}}(K_Z,w)$};\\
\nl ~~~ \textcolor{black}{Send $(\mathrm{\emph{sub-filterId}},stag,C_{id}, ind, C_{stag}, $ $\delta, \mu,$ $op)$ to the server};\\
\nl \textbf{else}\\
\nl ~~~$stag^{\text{-}1} = \mathrm{\textsf{$F_1$}}(K_T,w|| \emph{UpdateCnt}[w])$;\\
\nl ~~~$\emph{UpdateCnt}[w] \text{-}\text{-}$;\\ 
\nl ~~~$C^{\text{-}1}_{id} = \emph{TSet}[stag^{\text{-}1}]$    \tcc*{ocall} 
\nl ~~~$(id^{\text{-}1}||sort^{\text{-}1}_k) = C^{\text{-}1}_{id} \oplus \mathrm{\textsf{$F_2$}}(K_Z,w)$;\\
\nl ~~~$ind^{\text{-}1} \text{=} \mathrm{\textsf{$F_3$}}(K_X,w||id^{\text{-}1}),ind \text{=} \mathrm{\textsf{$F_3$}}(K_X,w||id)$;\\
\nl ~~~\textcolor{black}{$C_{stag} = \emph{ITSet}[ind]$} \tcc*{ocall}  
\nl ~~~\textcolor{black}{$(w||c) \text{=} C_{stag} \oplus \mathrm{\textsf{$F_2$}}(K_Z,w)$; \\ 
\nl ~~~$stag\text{=}\mathrm{\textsf{$F_1$}}(K_T,w||c)$};\\

\nl ~~~\textcolor{black}{Send $(\mathrm{\emph{sub-filterId}},stag,stag^{\text{-}1},ind, ind^{\text{-}1}, \delta,$ $ \mu,op)$ to the server;}\\
\underline{Server:}\\
\nl \textbf{if} $op=insert$ \textbf{then}\\
\nl ~~~$\emph{TSet}[stag] = C_{id}$, $\emph{ITSet}[ind] = C_{stag}$; \\
\nl ~~~$\emph{XSet}[\mathrm{\emph{sub-filterId}}].\mathrm{\textsf{Insert}}(\delta, \mu)$; \\
\nl \textbf{else}\\

\nl ~~~$\emph{TSet}[stag] = \emph{TSet}[stag^{\text{-}1}]$;\\
\nl ~~~$\emph{ITSet}[ind^{\text{-}1}] = \emph{ITSet}[ind]$\\
\nl ~~~Delete $\emph{ITSet}[ind]$;\\
\nl ~~~$\emph{XSet}[\mathrm{\emph{sub-filterId}}].\mathrm{\textsf{Delete}}(\delta, \mu)$;\\
\nl Update \emph{IndexTree} if necessary \tcc*{ecall}
\end{multicols}
% \DecMargin{2cm}

\end{algorithm}

\textsf{SetUp} \textbf{Protocol}. This protocol is responsible for initializing some secret keys and data structures. Specifically, the client launches a remote attestation and establishes a secure channel with the enclave first. Then it generates secret keys $(K_T, K_Z, K_X)$ for PRFs $(F_1, F_2, F_3)$ and sends them to the enclave. The enclave initializes three empty data structures: (1) \emph{UpdateCnt} that stores the number of $id_{in}$ corresponding to each $w$; (2) \emph{CFs} that stores a mapping from \emph{sub-filtersId} to cache \emph{sub-filters}; (3) \emph{IndexTree} that stores the split state of LDCF-encoded \emph{XSet} to locate the matching \emph{sub-filters} during the update and search procedures. The server initializes three empty data structures: (1) \emph{TSet} that stores the encrypted posting list; (2) \emph{ITSet} that stores a mapping from the neighbor vertices' $id$ to their location in the posting list to locate the deleting neighbor vertices during the deletion procedure; (3) \emph{XSet} that stores the fingerprints of all $xtag$. 

% \setlength{\textfloatsep}{0.1em}
% \begin{algorithm}[t]
% \scriptsize
% \caption{\textsf{Setup}}
% \label{alg:alg1}
% \SetNlSty{large}{}{:}
% % \SetAlgoNoLine
% % \SetAlgoNoEnd
% % \SetAlgoVlined
% % \LinesNumbered
% \KwIn{Security parameter $\lambda$.}
% \KwOut{Encrypted database \emph{EDB}$=(\mathrm{\textsf{TSet, ITSet, XSet}})$.}
% \underline{Client:}\\
% \nl $K_T,K_Z,K_X \leftarrow \{0,1\}^\lambda$;\\
% \nl $sk = (K_T, K_Z, K_X)$;\\
% \nl Launch a remote attestation and establish a secure channel;\\
% \nl Send $sk$ to the Enclave;\\
% \underline{Encalve:} \\
% \nl Initialize an update counter $\mathrm{\textsf{UpdateCnt}}$;\\
% \nl \textcolor{black}{Initialize a $\mathrm{\textsf{IndexTree}}$}; \\
% \underline{Server:} \\
% \nl Initialize an encrypted database \emph{EDB}$=(\mathrm{\textsf{TSet, ITSet, XSet}})$;\\
% \end{algorithm}

\begin{algorithm}[t]

\caption{\textsf{Search}}
\label{alg:alg3}\scriptsize
\SetNlSty{text}{}{:}
\SetAlgoNoLine
\SetAlgoNoEnd
% \SetAlgoVlined
% \LinesNumbered
\begin{multicols}{2}
\KwIn{Secrity keys $K_T, K_Z, K_X$, update counter \emph{UpdateCnt}, search token $(w_1 \wedge...\wedge w_n)$, encrypted database \emph{EDB}=(\emph{TSet, ITSet, XSet}).}
\KwOut{search result \emph{ResList}.}
\underline{Client:}\\
\nl Send $(w_1 \wedge...\wedge w_n)$ to the enclave;\\
\underline{Enclave:} \\
\nl Initialize two empty lists $stokenList$, $CResList$;\\
\nl \textbf{for} $j = 1$ \emph{\textbf{to}} \emph{UpdateCnt}$[w_1]$ \textbf{do} \\
\nl ~~~$stag = \mathrm{\textsf{$F_1$}}(K_T,w_1||j)$;\\
\nl ~~~Insert $stag$ into the $stokenList$;\\

\nl Send $stokenList$ to the server;\\

\underline{Server:}\\
\nl Initialize an empty list $CidList$;\\
\nl \textbf{for} $j = 1$ \emph{\textbf{to}} $stokenList.size$ \textbf{do} \\
\nl ~~~$C_{id_j} = \emph{TSet}[stokenList[j]]$;\\
\nl ~~~Insert $C_{id_j}$ into the $CidList$;\\
\nl Send $CidList$ to the enclave;\\

\underline{Enclave:}\\
\nl Initialize an empty map $CFs$;\\
\nl \textbf{for} {$\mathrm{\textsf{each }} C_{id_j} \in CidList$} \textbf{do}\\
\nl ~~~$flag = 0$;\\
\nl ~~~$(id_j||weight_j) = C_{id_j} \oplus \mathrm{\textsf{$F_2$}}(K_Z,w_1)$;\\
\nl ~~~\textbf{for} $i = 2$ \emph{\textbf{to}} $n$ \textbf{do} \\
\nl ~~~~~~$xtag_{i,j} = (w_i||id_j),\delta =H_2(xtag_{i,j})$;\\
\nl ~~~~~~Get \emph{sub-filterId} according to $\delta$ and \emph{IndexTree};\\
\nl ~~~~~~\textbf{if} \emph{sub-filterId} not in $CFs$ \textbf{then} \\
\nl ~~~~~~~~~Load matching \emph{sub-filter}\tcc*{ocall} 

\nl ~~~~~~~~~$CFs[\mathrm{\emph{sub-filterId}}] = \mathrm{\emph{sub-filter}}$;\\
\nl ~~~~~~$\mu = H_1(xtag_{i,j})$, $\nu = \mu \oplus H_1(\delta)$;\\
% \nl ~~~~~~\textbf{if} $\delta$ not in $CFs[\mathrm{\emph{sub-filterId}}][\mu]$ \testsf{and} $CFs[\mathrm{\emph{sub-filterId}}][\nu]$ \\
\nl ~~~~~~\textbf{if} $\delta$ not in $CFs$[\emph{sub-filterId}]$[\mu]$ \textsf{and} $CFs$[\emph{sub-filterId}]$[\nu]$ \\
\nl ~~~~~~~~~$flag = 1$, break;\\
\nl ~~~\If{$flag = 0$}{
\nl ~~~Insert $C_{id_j}$ into the $CResList$;\\   
}
\nl Send $CResList$ to the client;\\

\underline{Client:}\\
\nl Initialize an empty list $ResList$;\\
\nl \textbf{for} $\mathrm{\textsf{each }} C_{id_j} \in ResList$ \textbf{then} \\
\nl ~~~$(id_j|| \mathrm{\emph{$weight_j$}}) = C_{id_j} \oplus \mathrm{\textsf{$F_1$}}(K_T,w_1)$;\\
\nl ~~~Insert $id_j$ into the $ResList$;\\
\end{multicols}
\end{algorithm}
\textsf{Update} \textbf{Protocol}. This protocol allows the client to issue graph updates. At the beginning, the client generates an update request $(w, id, weight, op)$, where $op$ is the operation type (insertion or deletion) (line 1). Upon receiving it, the enclave first calculates $xtage$ using $w$ and $id$ and generates its fingerprint $\delta$ and candidate index $\mu$ (line 2). Then, the enclave searches the path of the matching \emph{sub-filter} \emph{sub-filterId} in \emph{IndexTree} based on $\delta$ (line 3). Specifically, the enclave traverses each bit of $\delta$ to deeply search the \emph{IndexTree} by selecting the left subtree when the bit is 0, otherwise selecting the right subtree, until finding a leaf node, and the traversed $\delta$ sub-string is the path of the \emph{sub-filter} \emph{sub-filterId}. If $op=insert$, the enclave increases the \emph{UpdateCnt}$[w]$ by 1, then encrypts the $(w||c, id)$, $(w||id,c)$ pairs in forms of $(stag,C_{id})$ and $(ind,C_{stag})$ for \emph{TSet} and \emph{ITSet} respectively. Finally, the enclave sends $(\delta, \mu, \mathrm{\emph{sub-filterId}})$ and the pairs to the server (lines 4-12). If $op=delete$, the enclave first calculates $w$'s latest $id$ index $stag^{-1} = F_1(K_T,w||\mathrm{\emph{UpdateCnt}}[w])$ and decreases \emph{UpdateCnt}$[w]$ by 1, then loads the latest encrypted $id$ (i.e., $C_{id}^{-1}$) and encrypts it (lines 14-17). Next, the enclave computes the indexes of the latest $id$ (i.e., $ind^{-1}$) and the deleting $id$ (i.e., $ind$) in \emph{ITSet} (line 18). Then, the enclave gets the deleting $id$'s encrypted index in \emph{TSet} (i.e., $C_{stag}$) from \emph{ITSet}$[ind]$ and decrypts it to $stag$ (lines 19-21). Finally, the enclave sends $(stag,stag^{-1},ind,ind^{-1},\delta,\mu,\mathrm{\emph{sub-filterId}})$ to the server (line 22). Upon receiving it from the enclave, if $op=insert$, the server will insert the $(stag, C_{id})$ and $(ind, C_{stag})$ pairs into \emph{TSet} and \emph{ITSet} respectively and insert the fingerprint $\delta$ into the \emph{XSet} (lines 23-25). If $op=delete$, the server will overwrite the deleting $id$ with the latest $id$ in \emph{TSet} and \emph{ITSet} then delete the fingerprint $\delta$ and deleting $ind$ in \emph{XSet} and \emph{ITSet} respectively to remove the corresponding $id$ (lines 26-30). Finally, if the LDCF-encoded \emph{XSet} is split, the server will execute an \emph{ecall} operation to update the \emph{IndexTree} (line 31).

\textsf{Search} \textbf{Protocol}. This protocol allows the client to issue conjunctive graph searches. Specifically, the client selects search keywords $(w_1 \wedge...\wedge w_n)$ and sends them to the enclave (line 1). Upon receiving them, the enclave first initializes two empty lists $stokenList$ and $CResList$ to store $w_1$'s $stag$ and final encrypted results respectively (line 2). Then, the enclave traverses \emph{UpdateCnt}$[w_1]$ to compute all of $stag$s for $w_1$ and sends them to the server (lines 3-6). The server initializes an empty list $CidList$, and then traverses the $stokenList$ to get the encrypted $id$ (i.e., $C_{id}$) of $w_1$ from \emph{TSet} and returns them to the enclave (lines 7-11). For each $C_{id_j}$, the enclave decrypts it to get $id_j$ and calcualtes the $xtag_{i,j}$ using $\{w_i\}, i \in [2,n]$  (lines 13-17). To hide the $xtag$'s fingerprint from the server, the enclave generates the fingerprint $\delta$ and finds the corresponding \emph{sub-filterId} from \emph{IndexTree}. If the matching \emph{sub-filter} doesn't exist in the cache table $CFs$, the enclave will execute an \emph{ocall} operation to load the matching \emph{sub-filter} (lines 18-21). After that, the enclave computes two candidate positions $\mu$ and $\nu$ to check whether the $\delta$ exists in the \emph{sub-filter}. Only all $xtag_{i,j},i \in [2,n]$ pass the membership check $id_j$ can be added into the final result $CResList$ (lines 19-26). Finally, the enclave returns $CResList$ to the client (line 27), and the client can obtain all decrypted $id_j$ in $ResList$ (lines 28-31).

\subsection{Extension to Fuzzy Search}
\emph{SecGraph} can also support the fuzzy search (such as sub-string search), e.g., find users whose \emph{name} contains 'Ha'. To enable this, instead of directly creating a $stag$ for each user's \emph{name}, our idea is to split the \emph{name} into a set of pairs composed of a sub-string with fixed-length $s$ and an integer. For example, assume the length of a sub-string is 2, 'Harry' can be split into
$\{(\text{'\#H'},1),(\text{'Ha'},2),(\text{'ar'},3),$ $(\text{'rr'},4),(\text{'ry'},5),(\text{'y\$'},6)\}$, 
where each integer part refers to the absolute position $pos$ of the sub-string in the \emph{name}, $\text{'\#'}$ is the start character and $\text{'\$'}s$ is the terminator. We only make a few changes to \emph{SecGraph} to enable it to support fuzzy search. Specifically, in $\small\mathrm{\textsf{Update}}$ protocol (Alg.\ref{alg:alg2}), the client sends $(w,id, pos,op)$ to the enclave, where $w$ is a sub-string (line 1). Then the enclave calculates $xtag=(w||id||pos)$ (line 2) and $C_{id}=(id||pos) \oplus \mathrm{\textsf{$F_2$}}(K_Z,w)$ (line 9). The posting list for each keyword $w$ is a set of $(id, pos)$ pairs, where $id$ is the vertex whose \emph{name} contains $w$ and $pos$ is the absolute position of $w$ in the \emph{name}. In $\small\mathrm{\textsf{Search}}$ protocol (Alg.\ref{alg:alg3}), the client sends the search keywords in the form of $w_1\wedge (w_2,\Delta_2)\wedge...\wedge (w_n,\Delta_i), i\in[2,n]$, where $\Delta_i$ is the relative position between $w_i$ and $w_1$ (line 1). Upon receiving the $CidList$ from the server, the enclave obtains $(id_j || pos_j) = C_{{id}_j} \oplus F_2(K_Z,w_1)$ (line 15). After that, for each $(w_i,\Delta_i)$, the enclave calculates the $xtag_{i,j}=(w_i||id_j||pos_j+\Delta_i)$ (line 17).

\section{Optimzation}
\noindent\textbf{\emph{SecGraph-G:} Fingerprint Grouping.} We observe that there exists a performance bottleneck in \emph{SecGraph} that much time may be spent on \emph{ocall} operations when loading numerous \emph{sub-filters} into the enclave during a search procedure. To solve the drawback, we propose an optimized scheme \emph{SecGraph-G} to reduce the number of \emph{sub-filters} to be loaded. The strategy is fingerprint grouping by generating the fingerprint of $w$ and the fingerprint of $xtag(w,id_{in})$ and then joining them together to form the final fingerprint instead of just generating the fingerprint of $xtag(w,id_{in})$. In this way, the fingerprints generated by a vertex and each of its adjacent vertices are highly likely to be stored in the same \emph{sub-filter} when constructing the LDCF-encoded \emph{XSet}.

\noindent\textbf{\emph{SecGraph-P:} Checking Parallelization.} Furthermore, we note that the current membership check process of $xtag$ is serial which also seriously affects the search performance of \emph{SecGraph} (lines 13-24 in Alg.\ref{alg:alg3}). With the observation that each membership check process is independent, we propose another optimized scheme \emph{SecGraph-P} to speed up the membership check process during the search procedure. The strategy is checking parallelization by loading all matching \emph{sub-filters} into the enclave first and then parallelizing to check whether each $xtag$ exists or not. In this way, \emph{SecGraph-P} only takes one membership check time to complete all required membership check processes, in the best-case scenario.

% Assuming each $w$ corresponds to 130 $id$, for a conjunctive search involving 10 $w$, the membership check needs to be performed a minimum of 130 times and a maximum of 1170 times. We observe that the membership check of $(w||id)$ pairs does not interfere with each other, which means the process can be parallelized. So we can parallelize this process to improve search performance.
\section{Security Analysis}
Before presenting a formal security analysis to show the security guarantee of \emph{SecGraph}, we first define the leakage functions and then use them
to prove the security. In $\small\mathrm{\textsf{Setup}}$ protocol, \emph{SecGraph} leaks nothing to the server except for the empty encrypted database \emph{EDB}. Thus we have $\mathcal{L}^{Stp} \text{=} ( |$\emph{TSet}$|,|$\emph{ITSet}$|,|$\emph{XSet}$|)$, where $|$\emph{TSet}$|,|$\emph{ITSet}$|$ and $|$\emph{XSet}$|$ are ciphertext lengths of data structures of \emph{TSet}, \emph{ITSet} and \emph{XSet} respectively. In $\small\mathrm{\textsf{Update}}$ protocol, \emph{SecGraph} leaks access on \emph{TSet}, \emph{ITSet} and \emph{XSet}. Thus, we have $\mathcal{L}^{Updt}\text{=} (op,|$\emph{TSet}$[stag]|,|$\emph{ITSet}$[ind]|,$ $|$\emph{XSet}$[sub\text{-}filterId]|)$, where $op=insert/delete$ denotes the update operation, \emph{TSet}$[stag]$ indicates the encrypted identifier to be inserted in \emph{TSet} with its $stag$, \emph{ITSet}$[ind]$ indicates the encrypted $stag$ to be inserted in \emph{ITSet} with its $ind$ and \emph{XSet}$[sub\text{-}filterId]$ indicates the fingerprint to be inserted in \emph{XSet} with its $sub\text{-}filterId$. In $\small\mathrm{\textsf{Search}}$ protocol, \emph{SecGraph} leaks the search token $stokenList$  and access pattern on \emph{TSet} when the server finds the matching entries in \emph{TSet} associated with $w_1$, defined as $\mathrm{ap}_{\text{\emph{TSet}}}$ and on \emph{XSet} when the enclave locates the desired \emph{sub-filters}, defined as $\mathrm{ap}_{\text{\emph{XSet}}}$. Thus, we have $\mathcal{L}^{Srch}\text{=}(stokenList,\mathrm{ap}_{\text{\emph{TSet}}},\mathrm{ap}_{\text{\emph{XSet}}}).$ Following the security definition in \cite{DBLP:conf/ndss/PatranabisM21}, we give the formal security definitions.

\begin{define}
	Let $\Pi \text{=} (\small\mathrm{\textsf{Setup}}$, $\small\mathrm{\textsf{Update}}$, $\mathrm{\small\textsf{Search}})$ be our SecGraph scheme. Consider the probabilistic experiments $\small\mathbf{Real}_{\mathcal{A}}(\lambda)$ and $\small\mathbf{Ideal}_{\mathcal{A},\mathcal{S}}(\lambda)$ with a probabilistic polynomial-time($\small\mathrm{\textsf{PPT}}$) adversary and a stateful simulator that gets the leakage function $\small\mathcal{L}$, where $\lambda$ is a security parameter. The leakage is parameterized by $\mathcal{L}^{Stp},\mathcal{L}^{Updt}$ and $\mathcal{L}^{Srch}$ depicting the information leaked to $\small\mathcal{A}$ in \mbox{each procedure.}
	
	$\mathbf{Real}_{\mathcal{A}}(\lambda)$: The challenger initialises necessary data structures by running $\small\mathrm{\textsf{Setup}}$. When inputting graphs chosen by $\small\mathcal{A}$, it makes a polynomial number of updates (i.e., addition and deletion). Accordingly, the challenger outputs the encrypted database EDB=(TSet, ITSet, XSet) with $\small\mathrm{\textsf{Update}}$ to $\small\mathcal{A}$. Then, $\small\mathcal{A}$ repeatedly performs graph searches. In response, the challenger runs $\small\mathrm{\textsf{Search}}$ to output the result to $\small\mathcal{A}$. Finally, $\small\mathcal{A}$ outputs a bit.
	
	$\mathbf{Ideal}_{\mathcal{A},\mathcal{S}}(\lambda)$: Upon inputting graphs chosen by $\small\mathcal{A}$, $\mathcal{S}$ initialises the data structures and creates encrypted database EDB=(TSet, ITSet, XSet) based on $\mathcal{L}^{Stp}$, and passes them to $\small\mathcal{A}$. Then, $\small\mathcal{A}$ repeatedly performs range queries. $\mathcal{S}$ simulates the search results by using $\mathcal{L}^{Updt}$ and $\mathcal{L}^{Srch}$ and returns them to $\small\mathcal{A}$. Finally, $\small\mathcal{A}$ outputs a bit.
	
	We say $\small\Pi$ is $\mathcal{L}$-adaptively-secure if for any $\small\mathrm{\textsf{PPT}}$ adversary $\small\mathcal{A}$, there exists a simulator $\mathcal{S}$ such that $|\mathrm{Pr}[\mathbf{Real}_{\mathcal{A}}(\lambda)\text{=}1]-\mathrm{Pr}[\mathbf{Ideal}_{\mathcal{A},\mathcal{S}}(\lambda)\text{=}1]| \leq negl(\lambda)$, where $negl(\lambda)$ denotes a negligible function in $\lambda$.
\end{define}

\begin{theorem}\label{thm:thm1} 
(Confidentiality of SecGraph). Assuming $(F_1,F_2,F_3)$ are secure PRFs and $(H_1,H_2)$ are secure hash functions. SecGraph is $\mathcal{L}$-secure against an adaptive adversary and ensures forward security and Type-III backward security\footnote{Forward security refers to newly inserted data is no longer linkable to searches issued before, and backward security refers to deleted data is no longer searchable in searches issued later \cite{DBLP:conf/ndss/PatranabisM21}.}.
\end{theorem}
\begin{proof}
We model the PRFs and the hash functions as random oracles $\{\mathcal{O}_{F_1},\mathcal{O}_{F_2},$ $\mathcal{O}_{F_3},\mathcal{O}_{H_1},\mathcal{O}_{H_2}\}$ and sketch the execution of the simulator $\mathcal{S}$. In $\mathrm{\textsf{Setup}}$ protocol, $\mathcal{S}$ simulates the encrypted database based on $\mathcal{L}^{Stp}$, which has the same size as the real one. Specifically, it includes two dictionaries $\mathcal{D}_1$ and $\mathcal{D}_2$ and a set $\mathcal{T}$. $\mathcal{S}$ further simulates the keys in the enclave by generating random strings ($\overline{k}_{1},\overline{k}_{2},\overline{k}_{3}$), which are indistinguishable from the real ones. When the first graph search sample $(w_1,w_2)$ is sent, $\mathcal{S}$ generates simulated tokens $\tilde{t}_i=\mathcal{O}_{F_1}(\tilde{k}_1||w_1||i)$ from $c$ to 1 and $c$ is the number of matched entries from $\mathcal{L}^{Stp}$. For each matching value $\alpha$ in $\mathcal{D}_1$ with the address $\tilde{t}_i$, another random oracles $\mathcal{O}_{F_2}$ is operated as $\tilde{R}=\mathcal{O}_{F_2}(\tilde{k}_2||w_1)\oplus \alpha$ to obtain $\tilde{R}$ inside, where $\tilde{R}$ has the same length as the real one. With $\tilde{k}_3$, three random strings $\delta = \mathcal{O}_{H_2}(\tilde{k}_3||w_2||\tilde{R}),\beta=\mathcal{O}_{H_1}(\delta)$ and $\gamma=\beta \oplus \sigma$ are calculated to check whether $\delta$ is in any one of the locations of the simulated set $\mathcal{T}$. If yes, $\mathcal{S}$ adds $\tilde{R}$ into the results. When a new triplet $(w,id,weight)$ is added, the results can also be simulated based on $\mathcal{L}^{Updt}$. Due to the pseudorandomness of PRFs and the hash function, $\mathcal{A}$ cannot distinguish between the tokens and results of \mbox{$\mathbf{Real}_{\mathcal{A}}(\lambda)$ and $\mathbf{Ideal}_{\mathcal{A},\mathcal{S}}(\lambda)$.}

Following the definitions of forward and backward security in \cite{DBLP:conf/ndss/PatranabisM21}, we prove \emph{SecGraph} achieves both forward security and Type-III backward security. Forward security is straightforward since the data structure \emph{UpdataCnt} ensures that $\small\mathcal{A}$ cannot generate search tokens to retrieve newly added identifiers when adding a new triplet. As for backward security, remembering when the entries in \emph{ITSet} were added and deleted leaks when additions and deletions for $w$ took place. Extracting all update counts and correlating them with the update timestamps reveals the specific addition that each deletion canceled. Nevertheless, the identifiers are encrypted by XORing a PRF value, the server cannot learn which identifiers contained $w$ but have not been removed. Based on these leakages, \emph{SecGraph} guarantees Type-III backward security.
\end{proof}

\section{Experimental Evaluation}

\begin{table}[t]
\setlength{\tabcolsep}{0.1mm}
\scriptsize
\setlength{\abovecaptionskip}{0.cm}
	\caption{Summary of the graph data used in our experiments.}
	\label{tab:tab1}
	\begin{tabular}{c|c|c|c|c|c}
		\hline 
		\textbf{Dataset}& \textbf{Nodes}  & \textbf{Edges} & \textbf{Edge type} & \textbf{Graph type} & \textbf{Source link} \\
		\hline
		Email&36,692&183,831& Friendship & Undirect&snap.stanford.edu/data/email-Enron.html\\
		Youtube&1,134,890&2,987,624& Exchange  & Undirect&snap.stanford.edu/data/com-Youtube.html\\
		Gplus&107,614&13,673,453& Share & Directed&snap.stanford.edu/data/ego-Gplus.html\\
		\hline
	\end{tabular}
% \vspace{-0.2em}
\end{table}

\subsection{Experiment Settings}
In the experiments, we implement the state-of-the-art encrypted graph search scheme PeGraph \cite{DBLP:journals/tifs/WangZJY22}, and our four schemes \emph{SecGraph}, \emph{SecGraph-G}, \mbox{\emph{SecGraph-P}} and \emph{SecGraph-A} (adopting both the fingerprint grouping and checking parallelization strategies) in about 5k LOCs of C++\footnote{Our code: https://github.com/XJTUOSV-SSEer/SecGraph.}. The client and server are deployed on a workstation equipped with an SGX-enabled Intel(R) Core(TM) i7-10700 CPU@2.60GHz with Ubuntu 18.04 server and 64GB RAM. For cryptographic primitives, we use the cryptography library Intel SGX SSL and OpenSSL (v1.1.1n) to implement the pseudorandom function via HMAC-256 and use SHA-256 to generate hash values for fingerprint. For implementing the LDCF-encoded \emph{XSet}, we adopt the open-source code of the LDCF\footnote{The code of LDCF: https://github.com/CGCL-codes/LDCF.} provided in \cite{DBLP:conf/icde/ZhangC0R21}. We use three real-world datasets Email, Youtube, and Gplus in our experiments, as shown in Table.\ref{tab:tab1}. All experiments were repeated 20 times and the average is reported.

% \begin{table}
% \setlength{\tabcolsep}{3.5mm}
% \scriptsize
% \setlength{\abovecaptionskip}{0.cm}
% 	\caption{Summary of the graph data used in our experiments.}
% 	\label{tab:tab1}
% 	\begin{tabular}{c|c|c|c|c}
% 		\hline 
% 		\textbf{Name}& \textbf{Nodes}  & \textbf{Edges} & \textbf{Edge Type} & \textbf{Graph Type} \\
% 		\hline
% 		Email-Enron&36,692&183,831& Friendship & Undirect\\
% 		com-Youtube&1134890&2987624& Communication & Undirect\\
% 		ego-Gplus&107614&13673453& Share & Directed\\
% 		\hline
% 	\end{tabular}
% % \vspace{-0.2em}
% \end{table}

\subsection{Performance Evaluation}
In our experiments, we default to setting the fingerprint length, the \emph{sub-filter} size, and bucket size in \emph{SecGraph} as 16 bits, 10,000, and 4, respectively, according to the recommendation of LDCF \cite{DBLP:conf/icde/ZhangC0R21} since there is an acceptable balance between the accuracy and speed of membership check under these parameters.

\noindent\textbf{Update Performance.} We first evaluate the insertion performance of PeGraph and \emph{SecGraph} as the used dataset size ratio increases from 20\% to 100\%. As we can see from Fig. \ref{exp:exp1}, the insertion time costs of the two schemes increase with the used dataset size ratio. It takes 6,128 ms, 103,790 ms, and 383,615 ms for \emph{SecGraph} to construct the encrypted database using the entire Email, Youtube, and Gplus datasets, respectively. It is worth noting that \emph{SecGraph} is considerably up to 58$\times$ faster than PeGraph. The reason is that PeGraph needs to perform a large number of expensive exponentiation modulo operations to calculate $xtag$ for each keyword-identifier pair to build \emph{XSet}, while \emph{SecGraph} only computes the same number of hash values due to the design of the LDCF-encoded \emph{XSet}. We further experiment that \emph{SecGraph} takes on average 0.36 ms and 0.37 ms to insert and delete a pair, respectively. The experiment results demonstrate that has better update performance than PeGraph.

\noindent\textbf{Search Performance.} Here, we evaluate the search performance of all considered schemes as the search keyword count increases from 2 to 10. Fig. \ref{exp:exp2} depicts that, for every dataset, \emph{SecGraph} takes less time to search compared with PeGraph and our optimized schemes perform better than \emph{SecGraph}. Specifically, \mbox{\emph{SecGraph}} yields up to 208$\times$ improvement in search time compared with \mbox{PeGraph}, and \emph{SecGraph-G}, \emph{SecGraph-P} and \emph{SecGraph-A} respectively yield up to 572$\times$, 3,331$\times$ and 3,430$\times$ improvements in search time compared with PeGraph. First, there are two reasons why the search procedure in \emph{SecGraph} is faster than that in PeGraph: (1) \emph{SecGraph} only requires one roundtrip to send search tokens to the server due to the design of proxy-token generation method, while PeGraph requires two roundtrips, causing the communication cost in PeGraph is up to 540$\times$ larger than that in \emph{SecGraph}. (2) \emph{SecGraph} generates one $xtag$ just by computing a hash value to check whether or not it exists due to the design of the LDCF-encoded \emph{XSet}, while PeGraph needs to execute an expensive exponentiation modulo operation to calculate a $xtag$. Then, \emph{SecGraph-G} performs better than \emph{SecGraph} since the former adopting the fingerprint grouping strategy reduces the number of \emph{ocall} (i.e., the number of loaded \emph{sub-filters}) by 96\% compared with the latter. \emph{SecGraph-P} performs better than \emph{SecGraph} since the former parallelizes the procedure of the membership check of the $xtag$.

\begin{figure*}[t]
\centering
\setlength{\abovecaptionskip}{0.cm}
\setcounter{subfigure}{0}
\subfigure[Email dataset]
{\includegraphics[width=0.32\linewidth]{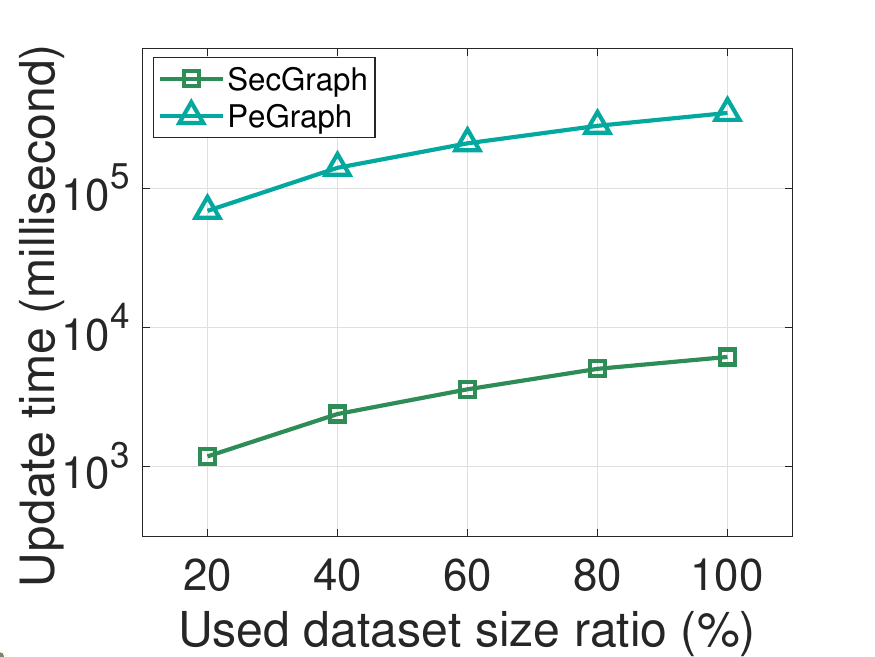}}
\subfigure[Youtube dataset]
{\includegraphics[width=0.32\linewidth]{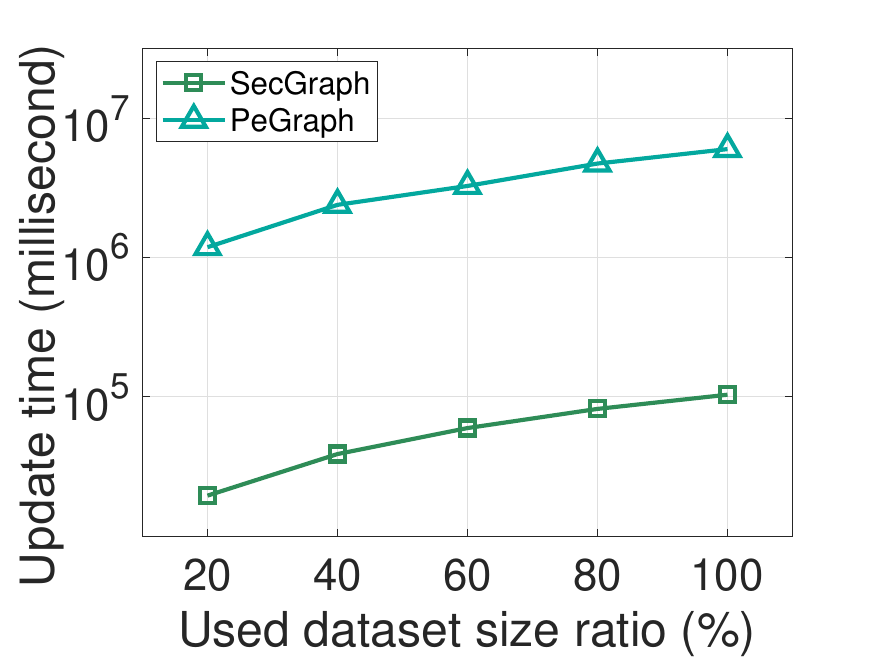}}
\subfigure[Gplus dataset]
{\includegraphics[width=0.32\linewidth]{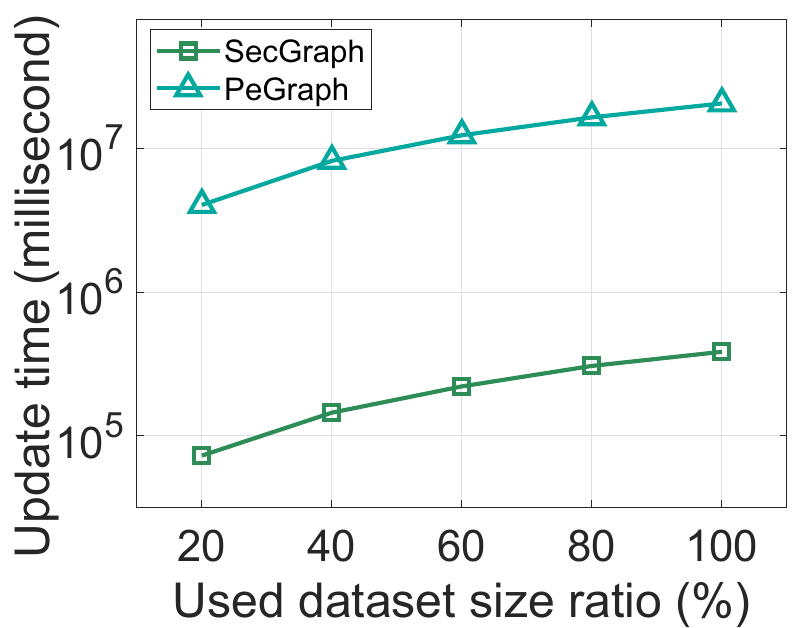}}
\renewcommand{\figurename}{Fig.}
\caption{Insertion performance in distinct datasets.}
\label{exp:exp1}
\vspace{-1em}
\end{figure*}

\begin{figure*}[t]
\centering
\setlength{\abovecaptionskip}{0.cm}
\setcounter{subfigure}{0}
\subfigure[Email dataset]
{\includegraphics[width=0.32\linewidth]{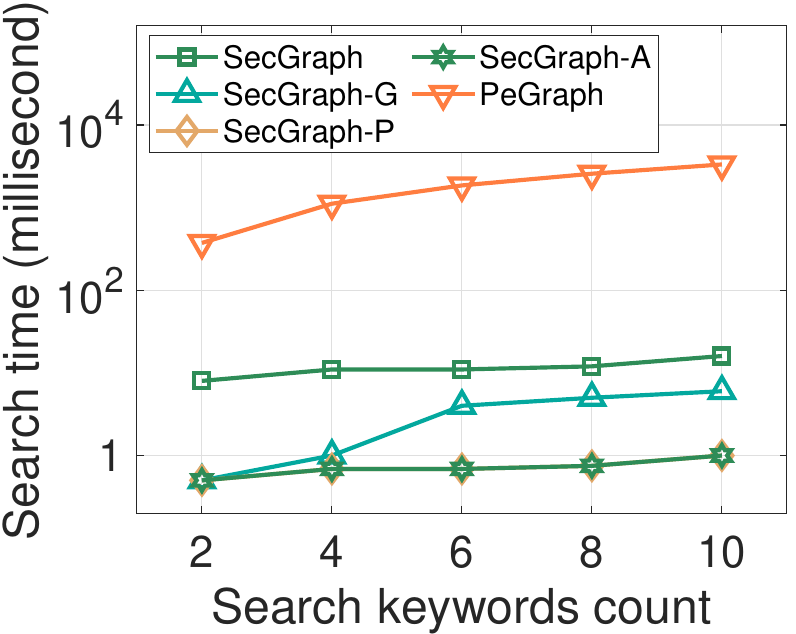}}
\subfigure[Youtube dataset]
{\includegraphics[width=0.32\linewidth]{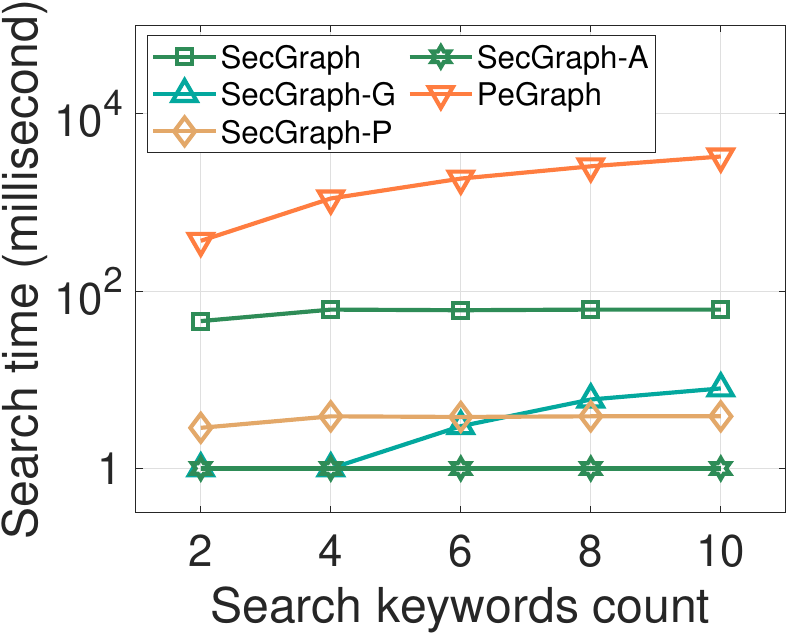}}
\subfigure[Gplus dataset]
{\includegraphics[width=0.32\linewidth]{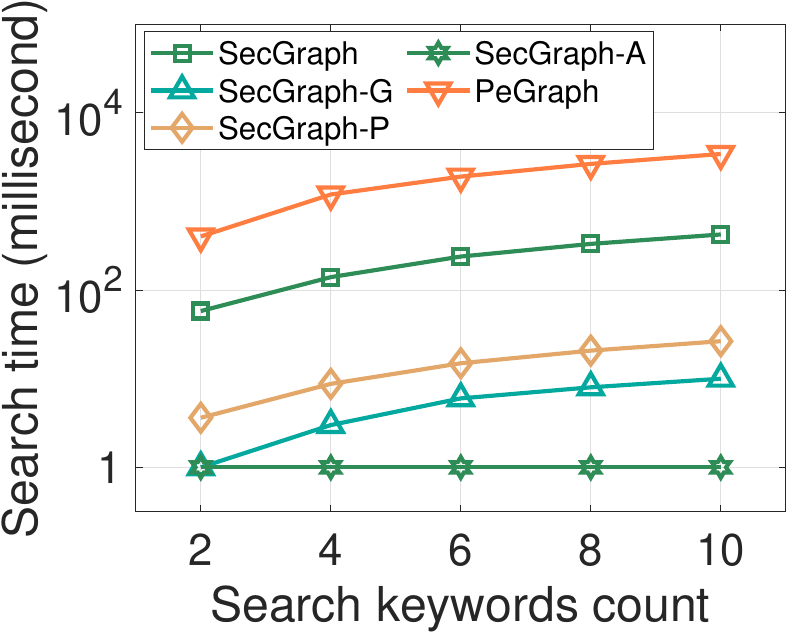}}
\renewcommand{\figurename}{Fig.}
\caption{Search performance in distinct datasets.}
\label{exp:exp2}
% \vspace{-1.6em}
\end{figure*}

\begin{figure*}[t]
\vspace{-1.5em}
\centering
\setlength{\abovecaptionskip}{0.cm}
\setcounter{subfigure}{0}
\subfigure[Email dataset]
{\includegraphics[width=0.32\linewidth]{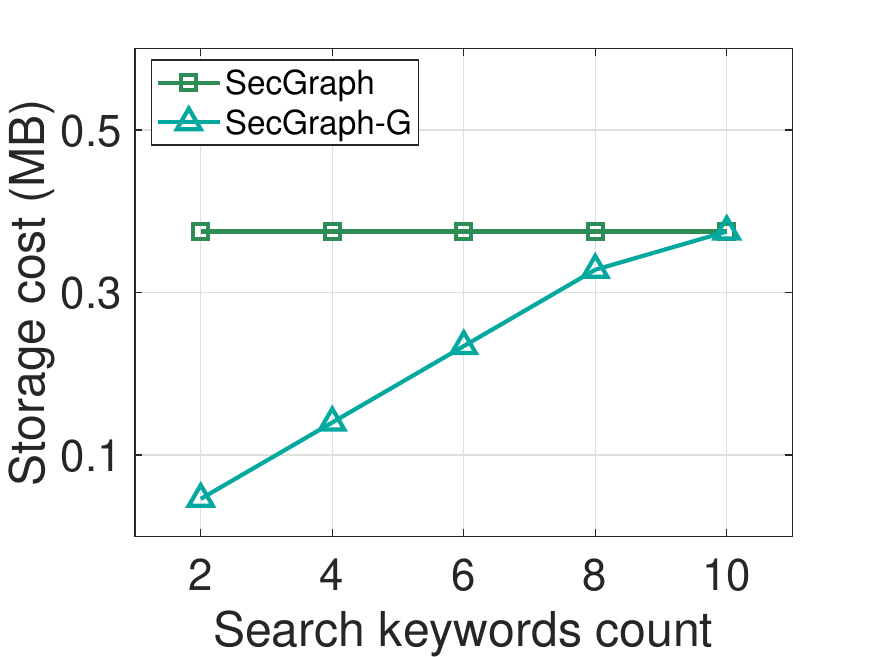}}
\subfigure[Youtube dataset]
{\includegraphics[width=0.315\linewidth]{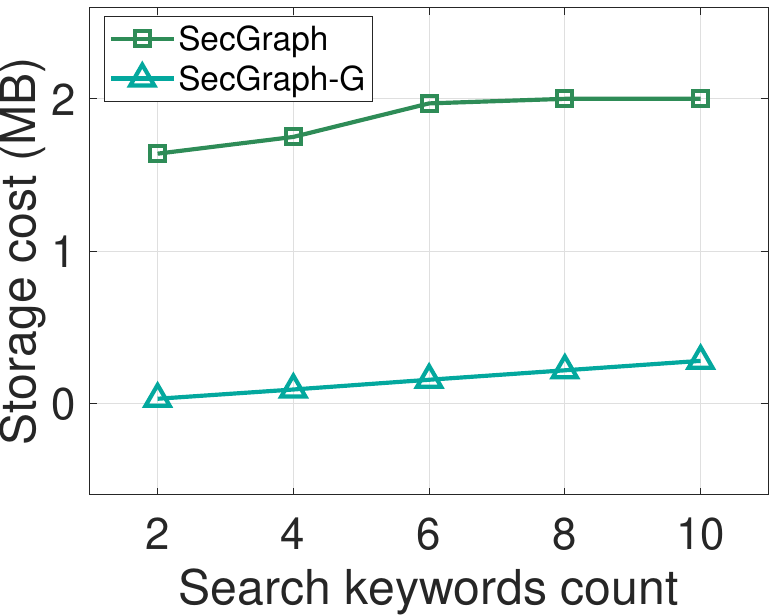}}
\subfigure[Gplus dataset]
{\includegraphics[width=0.32\linewidth]{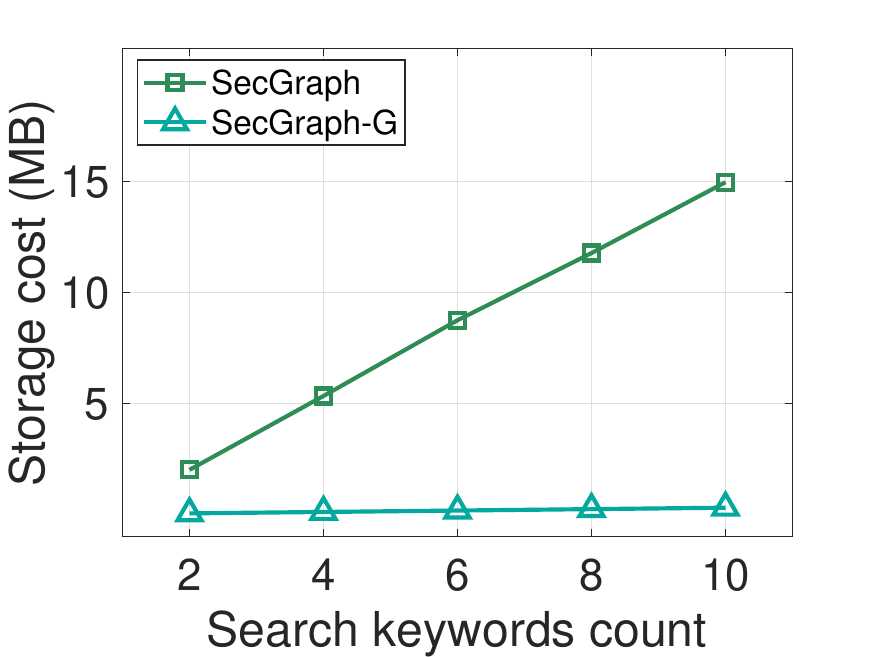}}
\renewcommand{\figurename}{Fig.}
\caption{Enclave storage cost in distinct datasets.}
\label{exp:exp3}
\vspace{-1em}
\end{figure*}
\begin{figure*}[t]
\centering
\setlength{\abovecaptionskip}{0.cm}
\setcounter{subfigure}{0}
\subfigure[Email dataset]
{\includegraphics[width=0.32\linewidth]{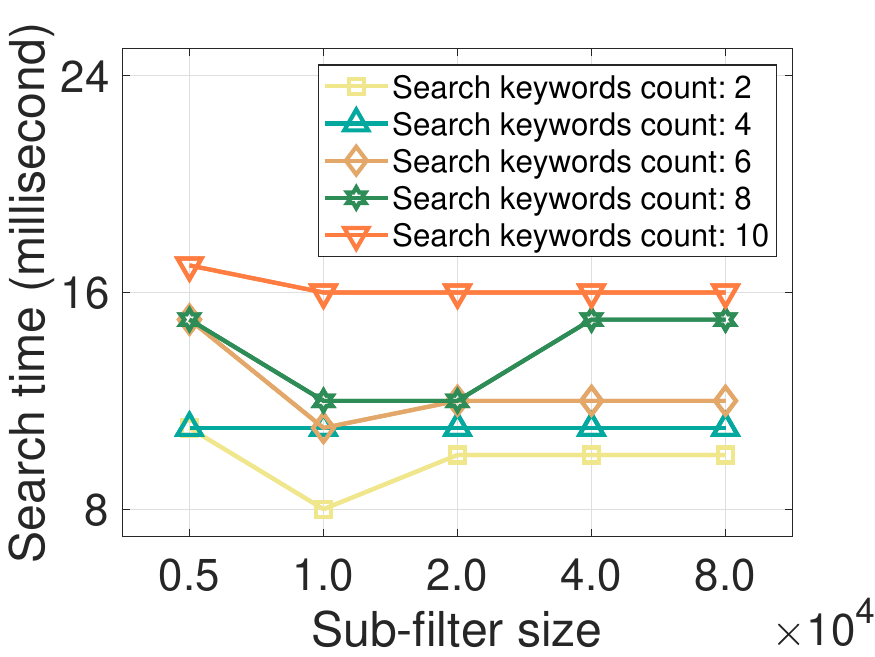}}
\subfigure[Youtube dataset]
{\includegraphics[width=0.32\linewidth]{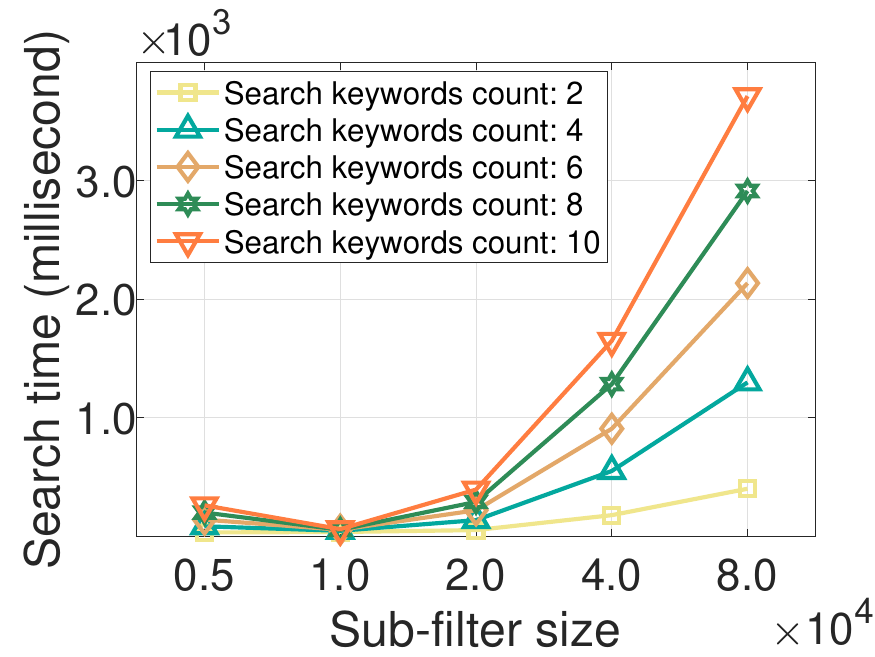}}
\subfigure[Gplus dataset]
{\includegraphics[width=0.32\linewidth]{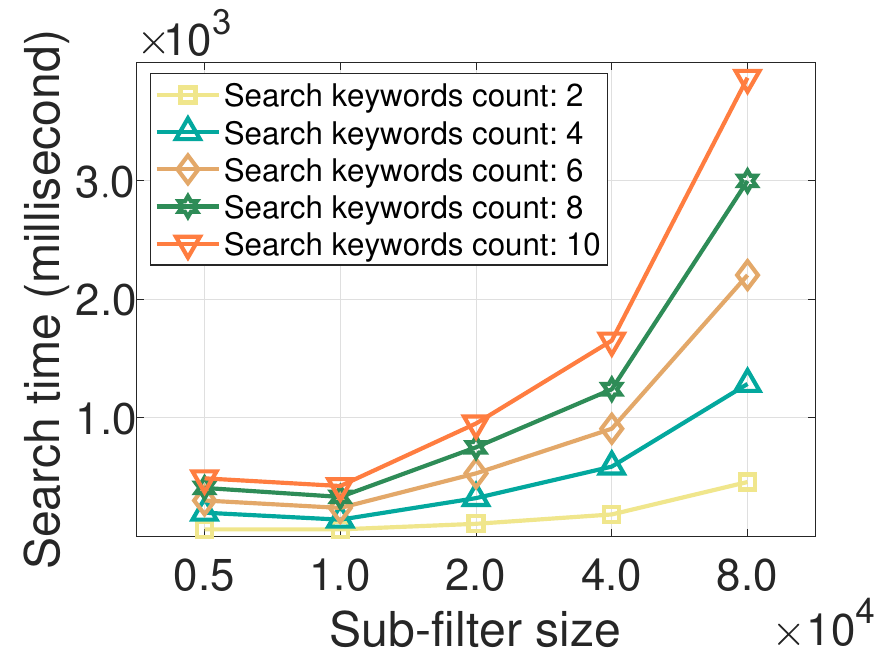}}
\renewcommand{\figurename}{Fig.}
\caption{Effect of the \emph{sub-filter} size.}
\label{exp:exp4}
% \vspace{-0.5em}
\end{figure*}
\begin{figure*}[t]
\centering
\setlength{\abovecaptionskip}{0.cm}
\setcounter{subfigure}{0}
\subfigure[Insertion time cost]
{\includegraphics[width=0.32\linewidth]{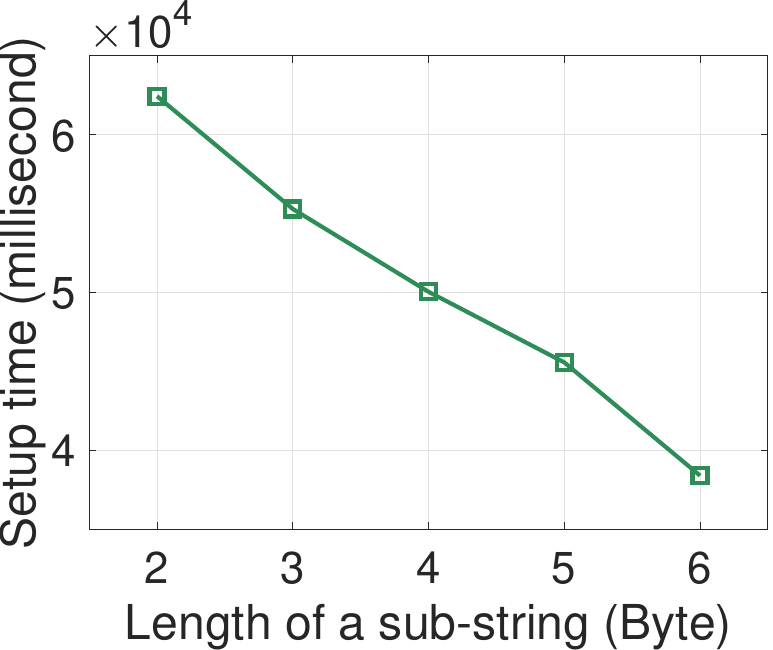}}
\subfigure[Storage cost]
{\includegraphics[width=0.32\linewidth]{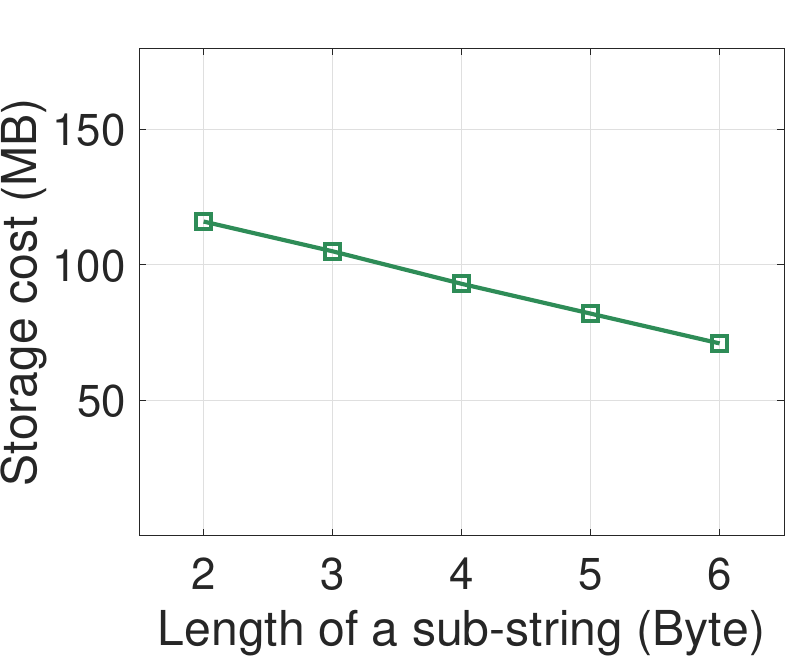}}
\subfigure[Search time cost]
{\includegraphics[width=0.32\linewidth]{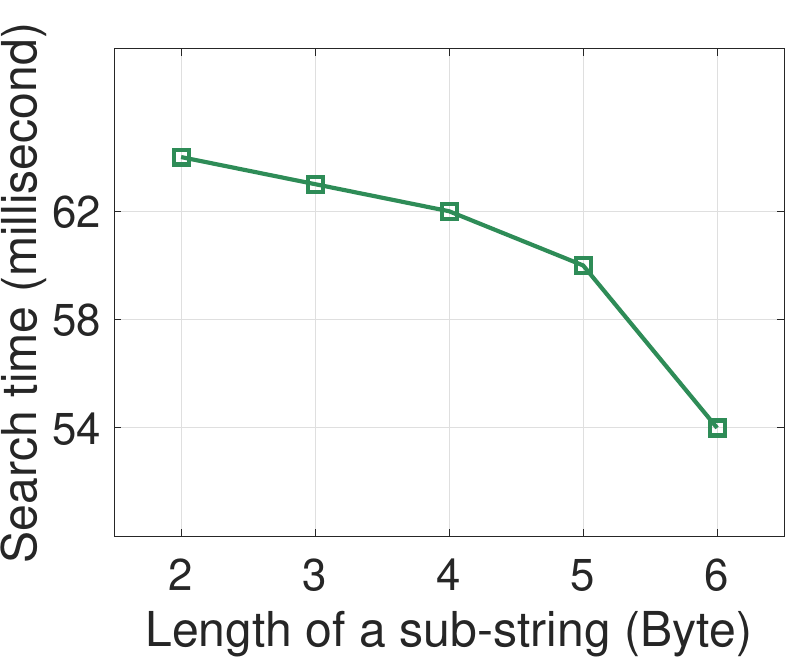}}
\renewcommand{\figurename}{Fig.}
\caption{Performance of fuzzy search enabled \emph{SecGraph}.}
\label{exp:exp5}
% \vspace{-0.5em}
\end{figure*}

Furthermore, we evaluate the enclave storage cost of \emph{SecGraph} and \emph{SecGraph-G} as the search keywords count increases from 2 to 10. As shown in Fig.\ref{exp:exp3}, we observe different results on datasets of different scales. For the small-size dataset (i.e., Email dataset), the enclave storage cost in \emph{SecGraph} hardly changes with the growth of search keywords count as there are only 20 \emph{sub-filters} in total and when the search keywords count is 2, all \emph{sub-filters} need to be loaded into the enclave. However, \emph{SecGraph-G} only requires loading 1 \emph{sub-filter} when the search keywords count is 2. For the middle-size dataset (i.e., Youtube dataset), the enclave storage cost in \emph{SecGraph} initially increases linearly with the search keywords count, until the search keywords count reaches 6, and the enclave storage cost in \emph{SecGraph} remains stable and unchanged. The reason is that when the search keywords count is 6, all \emph{sub-filters} need to be loaded into the enclave. But for \emph{SecGraph-G}, the enclave storage cost increases linearly with the search keywords count. For the large-size dataset (i.e., Gplus dataset), the enclave storage cost in \emph{SecGraph-G} also increases linearly with the search keywords count and there are 1408 \emph{sub-filters} in total and even when the search keywords count is 10, only 9 \emph{sub-filters} need to be loaded. Besides that, we surprisingly find that it only takes only 15 MB to store the loaded \emph{sub-filters} in the enclave.

\noindent\textbf{Effect of Parameters.} Next, we evaluate the search time cost in \emph{SecGraph} as the \emph{sub-filter} size (i.e., the number of fingerprints contained in a \emph{sub-filter}) increases from 5,000 to 80,000 under different search keywords count settings. Fig.\ref{exp:exp4} shows that, for every dataset, the search time cost is the lowest when the \emph{sub-filter} size is 10,000, which demonstrates the rationality of our default \emph{sub-filter} size setting. We analyze that setting the \emph{sub-filter} size too small will lead to a large number of \emph{sub-filters} being loaded during a search procedure, resulting in a high number of \emph{ocall}s and a long \emph{ocall} time. However, setting it too large can reduce the number of \emph{sub-filters} to be loaded appropriately, but it increases the amount of transferred data and also increases the \emph{ocall} time.

\noindent\textbf{Performance of Fuzzy Search Enabled \emph{SecGraph}.} 
Finally, we evaluate the performance of extended \emph{SecGraph} that supports fuzzy search as the sub-string length increases from 2 bytes to 6 bytes in the Email dataset. Specifically, as shown in Fig. \ref{exp:exp5}(a), the total insertion time cost in \emph{SecGraph} decreases as the sub-string length increases. This is reasonable since the shorter the sub-string length, the more sub-strings each vertex's \emph{name} is split into, which also means more keywords to be extracted, requiring more time to construct the encrypted database. Fig. \ref{exp:exp5}(b) shows that the storage cost of \emph{TSet} decreases linearly as the sub-string length increases due to the same reason mentioned above. Fig. \ref{exp:exp5}(c) shows that the search time slowly decreases as the sub-string length increases because the longer the sub-string length, the fewer matching results returned.

\section{Conclusion}
In this paper, we propose an SGX-based efficient and confidentiality-preserving graph search scheme \emph{SecGraph} to provide secure and efficient search services over encrypted graphs. Firstly, we design a new proxy-token generation method to reduce the communication cost. Then, we design an LDCF-encoded \emph{XSet} to reduce the computation cost. Moreover, we design a \emph{Twin-TSet} to enable encrypted search over dynamic graphs. We further extend \emph{SecGraph} into two optimized schemes \emph{SecGraph-G} and \emph{SecGraph-P} by adopting the fingerprint grouping and checking parallelization strategies, respectively, to speed up the search procedure. Finally, experiments and security analysis show that \emph{SecGraph} can achieve secure and efficient search over dynamic graphs. In the future, we will explore efficient and confidentiality-preserving graph search schemes that support more plentiful search services such as range search, boolean search, etc.

\bibliographystyle{splncs04}
\bibliography{reference.bib}
\end{document}